\documentclass[letterpaper, 10 pt, conference]{IEEEtran}
\usepackage[margin=0.75in]{geometry}
\usepackage{amsmath}
\usepackage{amsthm} %For proofs
\usepackage{amsfonts}
\usepackage{graphicx}
\graphicspath{{./images/FirstRevision/}}
\newtheorem{lemma}{Lemma}
\newtheorem{theorem}{Theorem}
\newtheorem{corollary}{Corollary}

\newtheorem{assumption}{Assumption}

\usepackage{caption}
\usepackage{subcaption} 
\usepackage{xcolor}
\usepackage{algpseudocode}
\usepackage{algorithm}
\usepackage{amssymb}
\usepackage{changes}
\usepackage[style=ieee,maxnames=2,minnames=1]{biblatex}
\AtEveryBibitem{%
  \clearfield{doi}%
  \clearfield{url}%
  \clearfield{isbn}%
  \clearfield{issn}%
}
\definecolor{kthblue}{rgb}{0.098, 0.329, 0.651}
\definecolor{newgreen}{RGB}{38, 115, 77}% used for general-purpose revisions.
\definecolor{newpurple}{RGB}{102, 0, 204}% used for new baselines and comparisons.
\definecolor{newred}{RGB}{153, 0, 0}% used for small improvements in citations 
\definecolor{newbrown}{RGB}{115, 77, 38}
\definecolor{neworange}{RGB}{255, 153, 0}
\usepackage{mathtools}

\title{
\vspace*{0.25in}%{3\baselineskip}
Minimal positive Markov realizations}
\author{Hamed Taghavian and Jens Sjölund}
\begin{document}
\maketitle

%\begin{keyword}                           % Five to ten keywords,  
%Optimal control; Linear-quadratic regulation; Symmetric systems
%\end{keyword}    

%%%%%%%%%%%%%%%%%%%%%%%%%%%%%%%%%%%%%%%%%%%%%%%%%%%%%%%%%%%%%%%%%%%%%%%%%%%%%%%%
\begin{abstract}
Finding a positive state-space realization with the minimum dimension for a given transfer function is an open problem in control theory. In this paper, we focus on positive realizations in Markov form and propose a linear programming approach that computes them with a minimum dimension. Such minimum dimension of positive Markov realizations is an upper bound of the minimal positive realization dimension. However, we show that these two dimensions are equal for certain systems.
\end{abstract}

%%%%%%%%%%%%%%%%%%%%%%%%%%%%%%%%%%%%%%%%%%%%%%%%%%%%%%%%%%%%%%%%%%%%%%%%%%%%%%%%
\section{Introduction}\label{sec:intro}
Positive systems are a family of linear time-invariant systems that generate non-negative states and output, when the initial states and input are non-negative. These systems appear in numerous fields, including chemistry, sociology, and electronics, where the state variables that describe the system are inherently non-negative, such as concentrations, population levels, and electric charge~\cite{poplevels,enzyme,taghavian2024model}.

A system is positive if and only if all its state-space matrices are non-negative. By definition, positive systems have non-negative impulse responses. However, not all transfer functions with non-negative impulse responses admit positive realizations; necessary and sufficient conditions for the existence of positive realizations are provided in~\cite{existence}. When these conditions are met, a positive realization exists for the transfer function and can be computed using standard algorithms~\cite{existence,existence&algorithm}. However, this realization may not have the minimum dimension possible. In fact, finding the minimal positive realization of a given transfer function is an open problem~\cite{survey}. Even the dimension $N^{\star}$ of such minimal positive realizations is unknown and can be larger than the order of the transfer function $n$. Except in special cases, only upper and lower bounds have been obtained for $N^{\star}$.%~\cite{survey}

A minimal positive realization is more economical and has computational benefits. For example, it reduces the size and power consumption of filters \cite{survey,CCD} and decreases the controller order in monotonic tracking systems~\cite{tac,ifac}. In general, the dimension of a positive realization directly affects the memory usage and computational complexity in optimal control systems~\cite{minimax}. Therefore, reducing the dimensions of positive state-space realizations has been the main concern of some recent research~\cite{survey,lam}.%\cite{reduce,lam}.

Markov realizations are a family of canonical state-space realizations in which the Markov parameters of the system appear explicitly in the state-space matrices~\cite[\S 9]{farinabook}. In this paper, we study the family of transfer functions that admit positive realizations in Markov form and show that the dimension of such realizations can be easily minimized using linear programs. The minimum dimension of such positive Markov realizations equals $N^{\star}$, \emph{i.e.}, the minimal positive realization dimension, in certain cases.

%The dimension of the positive realizations obtained in this method is generally an upper bound of $N^{\star}$. However, we show that these positive realizations are minimal in certain cases.

This paper is organized as follows. The preliminaries are provided in Section~\ref{sec:preliminaries}. We introduce the positive Markov realizations in Section~\ref{sec:Markov} and identify a family of systems that admit such realizations in Section~\ref{sec:realizable}. We evaluate the results for third-order systems in Section~\ref{sec:3rdorder} and give the conclusion in Section~\ref{sec:conc}. 

\section*{notation}
We use the notation $\mathbb{N}_0=\mathbb{N}\cup\lbrace 0\rbrace$. For $n,m\in\mathbb{N}$, we write $n\mid m$ if $n$ divides $m$, and $n\nmid m$ otherwise. The greatest common divisor of $n$ and $m$ is $\textnormal{gcd}(n,m)$ and the remainder of the Euclidean division of $n$ by $m$ is $n\mbox{ mod }m$. The element on the $i$th row and $j$th column of the matrix $A$ is $[A]_{ij}$ and the $i$th element of the vector $a$ is $[a]_i$. Inequalities are interpreted element-wise. For example, matrix $A$ is non-negative if $A\geq 0$, that is, all its elements are non-negative. When there is no confusion, we use the same notation $a$ for the vector $a=\begin{bmatrix}a_0\;a_1\;\dots \; a_N\end{bmatrix}^T\in\mathbb{R}^{N+1}$ and the sequence $a=\lbrace a_0,a_1,\dots,a_N \rbrace$, which is defined as the coefficient sequence of the polynomial $A(z)=a_0z^N+a_1z^{N-1}+\dots+a_{N-1}z+a_N$. The sequence $a$ is decreasing if $a_k\geq a_{k+1}$ for all $k=0,1,\dots,N-1$. We use $(a*b)_k=\sum_{j=0}^k a_jb_{k-j}$ to denote the convolution of the two sequences $a$ and $b$. The interior of the set $\mathcal{M}$ is $\textnormal{int}(\mathcal{M})$.\looseness=-1
% and use the notation $\mathcal{X}-\mathcal{M}$ for substraction of the two sets $\mathcal{X}$ and $\mathcal{M}$.

\section{Preliminaries}\label{sec:preliminaries}
Consider a linear single-input-single-output discrete-time system with the strictly proper transfer function 
\begin{align}\label{eqn:transfer_function}
H(z)=\frac{B(z)}{A(z)}&=\frac{b_1z^{n-1}+ \ldots +b_{n-1}z+b_n}{z^n+a_1z^{n-1}+ \ldots +a_{n-1}z+a_n}\nonumber\\
&=K\frac{ \prod_{k=1}^{m} (z-z_k)}{\prod_{k=1}^n (z-p_k)},
\end{align}
where there are no zero-pole cancellations and $b_k=0$ holds for $k\in[1,n-m)$. We assume the poles in (\ref{eqn:transfer_function}) are sorted in decreasing order of absolute values, that is
$$
\vert p_1\vert\geq \vert p_2\vert \geq \dots \geq \vert p_n\vert,
$$
and call $p_1$ a \emph{dominant} pole. We denote the Markov parameters of the system (\ref{eqn:transfer_function}) by $h_t=CA^tB$ where $t\in\mathbb{N}$. The state space realization
\begin{align*}
    \begin{array}{rl}
     x_{t+1}&=Ax_t+Bu_t,\\
    y_t&=Cx_t 
    \end{array}\quad t\in\mathbb{N}_0
\end{align*}
of the transfer function (\ref{eqn:transfer_function}) is called positive if all the state-space matrices $A$, $B$, $C$ are non-negative. The minimum dimension of such positive realizations is denoted by $N^\star$. We are interested in finding a finite upper bound $N$ for $N^{\star}$ and a state-space realization with dimension $N$. Therefore, without loss of generality, we make the following assumptions for the transfer function (\ref{eqn:transfer_function}) throughout the paper.

\begin{assumption}\label{ass:H}
The transfer function (\ref{eqn:transfer_function}) satisfies:
\begin{itemize}
    \item[1.] External positivity: $h_t\geq0$ for $t\in\mathbb{N}$.
    \item[2.] Normalized poles: $p_1=1$.
\end{itemize}
\end{assumption}

%System () has a non-negative impulse response.
The first assumption is not restrictive because every transfer function (\ref{eqn:transfer_function}) that admits a positive realization is externally positive, \emph{i.e.}, it has non-negative Markov parameters~\cite{tutorial}. Moreover, all these transfer functions have a non-negative dominant pole $p_1\geq 0$~\cite{cdc2}. When this dominant pole is zero, the system is known to have a minimal realization of order $n$~\cite{tutorial}. Therefore, we only consider the case where the system (\ref{eqn:transfer_function}) has a positive dominant pole $p_1>0$. 

%namely, System () has a dominant pole at unity.
The second assumption is not restrictive either because scaling the poles does not change the minimal realization dimension. To see this, let the strictly proper transfer function $H(z)$ have a positive dominant pole $p_1>0$ and the positive realization $(A,B,C)$ of dimension $N$. Then the transfer function
$$
H(z/p_1)= p_1 C(zI- p_1 A)^{-1}B
$$
has a positive realization $( p_1 A,B, p_1 C)$ of the same dimension $N$.

\subsection{Systems with a single positive pole}\label{sec:1pospole}
We use $\mathcal{X}$ to denote the set of all transfer functions~(\ref{eqn:transfer_function}) with $n$ poles and $m$ zeros that satisfy Assumption~\ref{ass:H}. In this paper, we pay special attention to the subset $\mathcal{M}\subseteq \mathcal{X}$ of transfer functions that have no positive poles other than $p_1=1$:
\begin{equation}\label{eqn:nopos}
    \mathcal{M}:=\lbrace H(z)\in\mathcal{X} \,\vert\, p_i\notin (0,+\infty),\, i\in[2,n] \rbrace.
\end{equation}
The systems in $\mathcal{M}$ may or may not have a positive realization. A system $H(z)\in\mathcal{M}$ has a positive realization if it does not have any other dominant poles than $p_1=1$ (\cite[Theorem 11]{tutorial}), or when all its poles have angles that are rational multiples of $\pi$ (\cite[Corollary 14]{tutorial}). A system $H(z)\in\mathcal{M}$ does \emph{not} have a positive realization if it has dominant poles with angles that are \emph{irrational} multiples of $\pi$, according to the Perron–Frobenius Theorem (see \emph{e.g.}, \cite[Theorem 4]{tutorial}).

% when the dominant poles have angles that are IRrational multiples of $\pi$ positive realization doesn't exist.
% when there are dominant poles that have angles that are rational multiples of $\pi$ & there are subdominant poles that have angles which are IRrational multiples of $\pi$, they may or may not have a positive realization.
%poles as $\exp(i\theta_j)$ satisfying
%$$\theta_j/\pi\in \mathbb{Q}$$

%This paper aims to find an upper bound $N$ for the minimal positive realization dimension $N^{\star}$ of the systems in the set (\ref{eqn:nopos}) when one exists. This upper bound is found by constructing a positive Markov realization of dimension $N$, introduced in the next section.

% N is based on the system poles $p$

%and write $H(z)\in\mathcal{M}$ when the transfer function (\ref{eqn:transfer_function}) satisfies (\ref{eqn:nopos}).
%Sufficient conditions for the existence of positive realizations: Single dominant pole and any angles. all rational multiples of pi. Necessary conditions for the existence of positive realization: dominants rational multiples of pi.
     
%The third assumption is restrictive because there is no algorithm known that can factorize a given system with several positive poles to a series connection of systems each having one positive pole and non-negative Markov parameters.

\section{Positive Markov realizations}\label{sec:Markov}
We introduce positive Markov realizations in this section and provide a linear program that computes these realizations when they exist. First, we extend the canonical Markov realization in \cite[\S 9]{farinabook} to higher dimensions, by introducing dummy zeros and poles to the transfer function (\ref{eqn:transfer_function}) as follows
\begin{align}\label{eqn:Htild}
    \tilde{H}(z)&=\frac{B(z)Q(z)}{A(z)Q(z)}\\
    &=\frac{B(z)Q(z)}{z^N+(a*q)_1z^{N-1}+ \ldots +(a*q)_N},\nonumber
\end{align}
where $Q(z)$ is a monic polynomial of order $N-n\in\mathbb{N}_0$. The transfer functions $H(z)$ and $\tilde{H}(z)$ are equal almost everywhere and have the same Markov parameters. Therefore, a realization of (\ref{eqn:transfer_function}) is
\begin{align}\label{eqn:Markov_realization}
    A&=\begin{bmatrix}
        0 & 0 & \dots & 0 & -(a*q)_N \\
        1 & 0 & \dots & 0 & -(a*q)_{N-1} \\
        0 & 1 & \dots & 0 & -(a*q)_{N-2} \\
        \vdots & \vdots & \ddots & \vdots & \vdots \\
        0 & 0 & \dots & 1 & -(a*q)_1 \\
    \end{bmatrix},\;
    B=\begin{bmatrix}
        1\\0\\\vdots\\0\\0
    \end{bmatrix}\\\nonumber
    C&=\begin{bmatrix}
        h_1&h_2&\dots &h_N
    \end{bmatrix},\; D=0,
\end{align}
%Therefore, when this realization is positive, the inequality (\ref{eqn:Nstar<N}) holds.
which has dimension $N(\geq n)$. The realization (\ref{eqn:Markov_realization}) is positive if and only if the last column of $A$ is non-negative, \emph{i.e.},
\begin{equation}\label{eqn:a*q<0}
    (a*q)_k\leq 0,\quad k\in[1,N].
\end{equation}
Condition (\ref{eqn:a*q<0}) is equivalent to the feasibility of the following linear program
    \begin{equation}\label{eqn:linprog}
    \begin{array}[c]{rl}
    \textnormal{find} & q\in \mathbb{R}^{N-n+1}\\
    \textnormal{subject to}
    & W\mathcal{T}_a q\leq 0\\
    & [q]_1=1,
    \end{array}
    \end{equation}
where the constant matrices $W\in\mathbb{R}^{N\times N+1}$ and $\mathcal{T}_a\in\mathbb{R}^{N+1\times N-n+1}$ are given by~\cite{ECC}
\begin{align*}
[\mathcal{T}_a]_{ij}=\left\lbrace\begin{array}{cc}
    1  &  i-j=0\\
    a_{i-j}  &  i-j\in[1,n]\\
    0 & i-j\not\in[0,n]
\end{array}\right.,\,     W=\begin{bmatrix}
    0& I
\end{bmatrix}.
\end{align*}
We conclude that a system has a positive Markov realization of dimension $N$, if and only if, the linear program (\ref{eqn:linprog}) is feasible with $N$. In this case, the minimal positive realization dimension $N^{\star}$ is upper bounded by $N$. The minimum $N$ that makes the linear program (\ref{eqn:linprog}) feasible is called the minimal positive \emph{Markov} realization dimension. This dimension is not always equal to $N^\star$, since a given system may have a positive realization of a lower dimension that is not of Markov type.
%Note that equality is not always met in (\ref{eqn:Nstar<N})

\section{Systems that admit positive Markov realizations}\label{sec:realizable}
In Section~\ref{sec:Markov}, we showed that when the linear program~(\ref{eqn:linprog}) is feasible, the transfer function (\ref{eqn:transfer_function}) has a positive realization on the Markov form (\ref{eqn:Markov_realization}) with dimension $N(\geq N^{\star})$. In this section, we study the feasibility of the linear program (\ref{eqn:linprog}). In particular, we show that the linear program (\ref{eqn:linprog}) is feasible for all the transfer functions in a dense subset of $\textnormal{int}(\mathcal{M})$ and infeasible for all transfer functions outside $\mathcal{M}$, where $\mathcal{M}$ is defined by (\ref{eqn:nopos}).

When the linear program (\ref{eqn:linprog}) is feasible with $N$, it is also feasible with $N+k$ for all $k\geq 0$. To see this, note that the non-zero elements in the coefficient sequences of $z^{k}Q(z)A(z)$ and $Q(z)A(z)$ are equal, rendering the condition (\ref{eqn:a*q<0}) unchanged. Therefore, if the linear program (\ref{eqn:linprog}) is infeasible for a given system, one may increase the realization dimension $N$ and try again.
%However, increasing $N$ puts pressure on computational resources as (\ref{eqn:linprog}) has $N-n$ variables and $N$ inequality constraints. In addition, the linear program (\ref{eqn:linprog}) is infeasible for some systems, regardless of how large $N$ is chosen. For example, when $$H(z)\not\in \mathcal{M}$$
However, the linear program (\ref{eqn:linprog}) is infeasible for some systems, regardless of how large $N$ is chosen. For example, when $$H(z)\not\in \mathcal{M},$$
the linear program (\ref{eqn:linprog}) is infeasible for all $N\in\mathbb{N}$. To see why, note that if $H(z)\not\in \mathcal{M}$ then $A(z)$ (and therefore $A(z)Q(z)$) has $k>1$ positive poles, and by Descartes' rule of signs, the sequence $a*q$ changes sign at least $k>1$ times, which makes it impossible to satisfy the condition (\ref{eqn:a*q<0}).

Therefore, to avoid an open-ended search, we identify a family of systems in $\mathcal{M}$ for which the linear program (\ref{eqn:linprog}) is feasible with a realization dimension $N=\mu_n$ obtained from the pole angles. Note that the dimension $\mu_n$ is often larger than the minimum possible for positive Markov realizations. However, it can be used in a bisection search to find the minimum positive Markov realization dimension, \emph{i.e.}, the smallest $N$ that makes the linear program (\ref{eqn:linprog}) feasible.

To present our result, we let $H(z)\in\mathcal{M}$ and define the polynomial $\hat{A}(z)$ as
$$
A(z):=(z-1)\hat{A}(z).
$$
The following lemma expresses the linear program (\ref{eqn:linprog}) feasibility in terms of the coefficients $\hat{a}$ of the polynomial $\hat{A}(z)$.

\begin{lemma}\label{lem:AhatQ}
    Let $H(z)\in\mathcal{M}$. The linear program (\ref{eqn:linprog}) is feasible if and only if the polynomial $\hat{A}(z)Q(z)$ has non-negative decreasing coefficients.
\end{lemma}
\begin{proof}
    The linear program (\ref{eqn:linprog}) is feasible if and only if condition (\ref{eqn:a*q<0}) is satisfied, which is equivalent to
    $$
    (\hat{a}*q)_k=(\hat{a}*q)_k-(\hat{a}*q)_{k-1}\leq 0,\quad k\in[1,N].
    $$
    Expanding this condition for $k=1,2,\dots,N$ reads
    $$
    1=(\hat{a}*q)_0\geq \dots \geq (\hat{a}*q)_{N-1} \geq 0.
    $$
\end{proof}

%Let us denote the system poles in (\ref{eqn:transfer_function}) with non-negative imaginary parts by
%\begin{equation}\label{eqn:rcos(teta)_np}
%    r_k\exp(i\theta_k),\quad k=1,2,\dots,n_p
%\end{equation}
%where $1\leq n_p\leq n$. By Assumption~\ref{ass:H}, we have $r_1=1$ and $\theta_1=0$. The next theorem provides an upper bound for the minimum positive Markov realization dimension of the systems in $\mathcal{M}$ when $\theta_k/\pi\in \mathbb{Q}$ for all $k=1,2,\dots,n_p$. In these systems, one can write the angle of the poles in (\ref{eqn:rcos(teta)_np}) as follows
%\begin{equation}\label{eqn:angles:n/m}
%    \theta_k=\left\lbrace\begin{array}{ll}
%       0  &  k=1\\
%       2\pi {l_k}/{m_k}  & k\in[2,n_p]
%    \end{array}\right.
%\end{equation}
%where $m_k,l_k\in\mathbb{N}_0$ are relatively prime and $m_k>l_k$.

The following theorem indicates that the linear program (\ref{eqn:linprog}) is feasible for a family of systems in $\mathcal{M}$ whose pole angles are rational multiples of $\pi$.

%Define the sequence
%    \begin{equation}\label{eqn:mu}
%        \mu_k=\left\lbrace\begin{array}{ll}
%            1, & k=1\\
%           m_k\mu_{k-1},  & 2\leq k\leq n_p
%        \end{array}\right.
%    \end{equation}
%\begin{equation}\label{eqn:angles:n/m}
%    \theta_k=\left\lbrace\begin{array}{ll}
%       0  &  k=1\\
%       2\pi {l_k}/{m_k}  & k\in[2,n_p]
%    \end{array}\right.
%\end{equation}

\begin{theorem}\label{thm:mu}
Let $H(z)\in\mathcal{M}$ and
\begin{align}\label{eqn:rcos(teta)_np}
    r_k\exp(i 2\pi {l_k}/{m_k}),\quad k=1,2,\dots,n_p
\end{align}
be the poles of $H(z)$ with non-negative imaginary parts sorted in decreasing magnitudes, \emph{i.e.}, $1=r_1\geq r_2\geq \dots \geq r_{n_p}$ where $ n_p\in[1, n]$. Assume $m_1=1,l_1=0$, that $m_k,l_k\in\mathbb{N}$ are relatively prime and $m_k>l_k$ for $k\in[2,n_p]$. Then the linear program (\ref{eqn:linprog}) is feasible with $N=\prod_{k=1}^{n_p}m_k$ if
\begin{equation}\label{eqn:m_notod_mu}
        m_k\nmid m_1m_2\dots m_{k-1}, \quad k\in[2,n_p].
    \end{equation}
\end{theorem}
\begin{proof}
The case $n_p=1$ is trivial as $\hat{A}(z)=1$, and therefore, the condition (\ref{eqn:a*q<0}) is satisfied by $Q(z)=1$ (Lemma~\ref{lem:AhatQ}) and the linear program (\ref{eqn:linprog}) is feasible with $N=1$. Therefore, we only consider the case $n_p\geq 2$. Before beginning the proof, we first define the sequence
    \begin{equation}\label{eqn:mu}
        \mu_k:=\left\lbrace\begin{array}{ll}
            1, & k=1\\
           m_k\mu_{k-1},  & 2\leq k\leq n_p
        \end{array}\right.
    \end{equation}
and the polynomial functions
\begin{align}\label{eqn:P&Omega}
P_j(z)&=z^{m_j-1}+r_j^{\mu_{j-1}} z^{m_j-2}+\dots +r_j^{\mu_{j-1}(m_j-2)}z\nonumber\\
&+r_j^{\mu_{j-1}(m_j-1)},\nonumber\\
\Omega_k(z)&=\prod\nolimits_{j=2}^{k} P_j(z^{\mu_{j-1}}), \quad k=2,3,\dots,n_p.
\end{align}
%Q_j(z)=P_j(z^{\mu_{j-1}}),\quad j=2,3,\dots,n_p\\

We prove this theorem by showing that the polynomial $\Omega_{n_p}(z)$ contains all the roots of $\hat{A}(z)$ (and hence it is divisible by $\hat{A}(z)$) and that it has non-negative decreasing coefficients. Then Lemma~\ref{lem:AhatQ} asserts that the coefficients of $Q(z)=\Omega_{n_p}(z)/\hat{A}(z)$ satisfy the linear constraints in (\ref{eqn:linprog}) and thereby, the linear program (\ref{eqn:linprog}) is feasible.

To prove that $\Omega_{n_p}(z)$ contains all the roots of $\hat{A}(z)$, note that $P_j(z)=0$ if and only if $z$ is an $m_j$-th root of $r_j^{\mu_{j-1}m_j}$ in the region $z\in\mathbb{C}-(0,+\infty)$.
%we first note that the roots of $P_j(z)$ in (\ref{eqn:P&Omega}) are $z_{k'}=r_j^{\mu_{j-1}}\exp(i\gamma_{j,k'})$ where
%$$
%\gamma_{j,k'}=\frac{2k'\pi}{m_j}, \quad k'\in(0,m_j).
%$$
Therefore, the roots of the polynomial $P_j(z^{\mu_{j-1}})$ are given by $z=r_j\exp(i\phi_{j,k',l'})$, where
$$
\phi_{j,k',l'}=\frac{2k'\pi}{m_j\mu_{j-1}}+\frac{2l'\pi}{\mu_{j-1}}, \;
l'\in[0,\mu_{j-1}), k'\in(0,m_j).
$$
%y (\ref{eqn:m_notod_mu})
Let $\theta_j=2\pi l_j/m_j$ in which $j\in[ 2,n_p]$. Since $\theta_j,\phi_{j,k',l'}\in[0,2\pi]$, the polynomial $P_j(z^{\mu_{j-1}})$ contains $r_j\exp(i\theta_j)$ (and its complex-conjugate when $\theta_j\neq \pi$) among its roots if there are $l'$ and $k'$ such that $\phi_{j,k',l'}=\theta_j$, which is equivalent to
$$
\mu_{j-1}l_j=k'+m_jl'.
$$
The above equation has a solution for $l'\in[0,\mu_{j-1})$ and $k'\in(0,m_j)$ if and only if
$$
m_j\nmid \mu_{j-1}l_j.
$$
This condition is equivalent to (\ref{eqn:m_notod_mu}) since $\textnormal{gcd}(m_j,l_j)=1$. Therefore, all the roots of $\hat{A}(z)$ are included in the roots of the polynomial $\Omega_{n_p}(z)$.

Next, we prove that the polynomial $\Omega_k(z)$ in (\ref{eqn:P&Omega}) has non-negative decreasing coefficients $\omega^{(k)}=\lbrace \omega^{(k)}_t\rbrace$ for all $k=2,3,\dots,n_p$. We use induction for this purpose. For $k=2$, we have
\begin{align*}
    \Omega_2(z)&=P_2(z)=z^{\mu_2-1}+r_2 z^{\mu_2-2}+\dots +r_2^{\mu_2-2}z\\
    &+r_2^{\mu_{2}-1},
\end{align*}
which has non-negative decreasing coefficients as $r_2\in[0,1]$. To prove this point for $k\in[3,n]$, we assume the polynomial
\begin{align*}
    \Omega_{k-1}(z)&=z^{\mu_{k-1}-1}+ \omega^{(k-1)}_1 z^{\mu_{k-1}-2}+\dots\\
    &+\omega^{(k-1)}_{\mu_{k-1}-2} z+\omega^{(k-1)}_{\mu_{k-1}-1}
\end{align*}
has non-negative decreasing coefficients. Let $\rho^{(j)}=\lbrace \rho^{(j)}_t\rbrace$ denote the coefficient sequence of the polynomial $P_j(z^{\mu_{j-1}})$ where
$$
\rho^{(j)}_t=\left\lbrace \begin{array}{ll}
     r^t_j, & t\in[0,\mu_j-\mu_{j-1}],\,\mu_{j-1}\mid t \\
     0, & \textnormal{otherwise} 
\end{array}\right. .
$$
Then the coefficients of $\Omega_k(z)$ are given by $\omega^{(k)}=\omega^{(k-1)}*\rho^{(k)}$ as follows
\begin{equation}\label{eqn:omega_k,t}
    \omega^{(k)}_{t}=\left\lbrace\begin{array}{ll}
        r_k^{s(t) \mu_{k-1}}\omega^{(k-1)}_{\tau(t)}, & t\in[0,\mu_k-1] \\
        0, & \textnormal{otherwise}
         \end{array}\right. ,
\end{equation}
where $s(t)=\lfloor t/\mu_{k-1} \rfloor$ and $\tau(t)=t\mbox{ mod }\mu_{k-1}$. Since the sequence $\omega^{(k-1)}$ is decreasing by assumption, the sequence~(\ref{eqn:omega_k,t}) is also decreasing whenever $s(t)$ is constant, that is, in the intervals
$$
t\in[s\mu_{k-1},(s+1)\mu_{k-1}), \quad s=0,1,\dots,m_k-1.
$$
Thus, the sequence (\ref{eqn:omega_k,t}) is decreasing in the whole range $t\in[0,\mu_k-1]$, if and only if $\omega^{(k)}_{t}\leq \omega^{(k)}_{t-1}$ also holds when $\mu_{k-1}\mid t$. This condition is equivalent to
\begin{align}\label{step1}
    r_k^{\mu_{k-1}} &\leq \omega^{(k-1)}_{\mu_{k-1}-1} \\
    &=\prod_{j=2}^{k-1} r_{j}^{\mu_j-\mu_{j-1}}.\nonumber
\end{align}
However, since the sequence $\lbrace r_j\rbrace_{j=2}^k$ is decreasing, we have
$$
    \prod_{j=2}^{k-1} r_{j}^{\mu_j-\mu_{j-1}}\geq
\prod_{j=2}^{k-1} r_{k}^{\mu_j-\mu_{j-1}}= 
r_{k}^{\mu_{k-1}-1}.
$$
Therefore inequality (\ref{step1}) is satisfied if $r_{k}^{\mu_{k-1}}\leq r_{k}^{\mu_{k-1}-1}$, which is true as $r_k\in[0,1]$. This proves that $\omega^{(k)}$ is a non-negative and decreasing sequence. As this result holds for all $k\in[2,n_p]$, the coefficients of $\Omega_{k}(z)$ are non-negative and decreasing.

Since the polynomial $\Omega_{n_p}(z)$ contains all the roots of $\hat{A}(z)$ and has non-negative decreasing coefficients, the linear program (\ref{eqn:linprog}) is feasible with
\begin{align}\label{eqn:order1}
N&=\textnormal{Deg}(Q)+n\nonumber\\
&=\textnormal{Deg}(\Omega_{n_p})-\textnormal{Deg}(\hat{A})+n\nonumber\\
&=\textnormal{Deg}(\Omega_{n_p})+1,
\end{align}
according to Lemma~\ref{lem:AhatQ}. To calculate $\textnormal{Deg}(\Omega_{n_p})$ in (\ref{eqn:order1}), we first use the relation (\ref{eqn:mu}) to eliminate $m_j$ in (\ref{eqn:P&Omega}) and write the polynomial $P_j(z^{\mu_{j-1}})$ in the following form
\begin{align*}
P_j(z^{\mu_{j-1}})&=z^{\mu_j-\mu_{j-1}}+r_j^{\mu_{j-1}}z^{\mu_j-2\mu_{j-1}}+\dots\\
&+r_j^{\mu_j-2\mu_{j-1}}z^{\mu_{j-1}}+r_j^{\mu_{j}-\mu_{j-1}},
\end{align*}
where $2\leq j\leq n_p$. Then, the degree of the polynomial $\Omega_k(z)$ in (\ref{eqn:P&Omega}) can be written as
\begin{align}\label{eqn:degOmega}
\textnormal{Deg}(\Omega_k)&=\sum\nolimits_{j=2}^{k}\textnormal{Deg}(P_j(z^{\mu_{j-1}}))\nonumber\\  
&=\sum\nolimits_{j=2}^{k} \mu_j-\mu_{j-1}\nonumber\\
&=\mu_{k}-1.
\end{align}
By substituting (\ref{eqn:degOmega}) in (\ref{eqn:order1}), we conclude that the linear program (\ref{eqn:linprog}) is feasible with
$
N=\mu_{n_p}=\prod_{k=1}^{n_p}m_k
$.
\end{proof}

The following corollary is a consequence of Theorem~\ref{thm:mu}.

\begin{corollary}\label{cor:intM}
    The set of systems for which the linear program (\ref{eqn:linprog}) is feasible is dense in $\textnormal{int}(\mathcal{M})$.
\end{corollary}
\begin{proof}
%Let $\mathcal{Q}$ denote the set of systems whose pole angles are rational multiples of $\pi$ as (\ref{eqn:angles:n/m}) and satisfy (\ref{eqn:m_notod_mu}).
To prove this result, we show that for all $H(z)\in
\textnormal{int}(\mathcal{M})$, there is an arbitrarily close transfer function
\begin{equation}\label{eqn:G}
    G(s)=K\frac{\prod_{k=1}^m (z-z''_k)}{\prod_{k=1}^n (z-p''_k)}=\frac{B(z)}{\prod_{k=1}^n (z-p''_k)}
\end{equation}
to $H(z)$ in the sense of the metric
\begin{equation}\label{eqn:metr}
    d(H,G):=\sum_{k=1}^m \vert z''_k-z_k\vert+\sum_{k=1}^n \vert p''_k-p_k\vert,
\end{equation}
for which the linear program (\ref{eqn:linprog}) is feasible. Hence the set of systems that make the linear program (\ref{eqn:linprog}) feasible is dense in the topological space induced by (\ref{eqn:metr}). The only difference between the transfer functions (\ref{eqn:G}) and (\ref{eqn:transfer_function}) is their sets of poles. The poles of (\ref{eqn:G}) are perturbed enough from $p$ to satisfy the hypothesis of Theorem~\ref{thm:mu}: to have angles that are rational multiples of $\pi$ and to satisfy the condition (\ref{eqn:m_notod_mu}).

To obtain such poles $p''$ for (\ref{eqn:G}), we first find angles that are rational multiples of $\pi$. Assume $r_k\exp(i \theta)$ are the poles of (\ref{eqn:G}) with non-negative imaginary parts sorted in decreasing magnitudes $r_k$ where $k=1,2,\dots,n_p$. Define $l'_1=0,m'_1=1$ and choose $l'_k,m'_k\in\mathbb{N}$ for $k\in[2,n_p]$ such that $l'_k<m'_k$ and $\textnormal{gcd}(m'_k,l'_k)=1$, where
\begin{equation}\label{eqn:l'm'e'}
    2\pi l'_k/m'_k=\theta_k+\varepsilon'_k, \quad k=1,2,\dots,n_p.
\end{equation}
%$p'_k=r_k\exp(i2\pi l'_k/m'_k)$
When $\theta_k/2\pi \in \mathbb{Q}$, we can choose $l'_k$, $m'_k$ such that $\varepsilon'_k=0$ holds in (\ref{eqn:l'm'e'}). Otherwise, since $\mathbb{Q}$ is dense in $\mathbb{R}$, one may choose $l'_k$, $m'_k$ such that $\vert\varepsilon'_k\vert$ is as small as desired.

Next, we use the following procedure that takes $l',m'$ and gives $l'',m''$ to satisfy the following analogue of condition~(\ref{eqn:m_notod_mu})
\begin{equation}\label{eqn:m_notod_mu''}
        m''_k\nmid m''_1m''_2\dots m''_{k-1}, \quad k\in[2,n_p].
    \end{equation}
\begin{itemize}
    \item[1.] Let $k=1$, $m''_1=1$, $l''_1=0$, $\mu''_1=1$.
    \item[2.] Set $k \gets k+1$. 
    \item[3.] If $m'_k \nmid \mu''_{k-1}$, set $m''_k\gets m'_k$ and $l''_k\gets l'_k$.
    \item[4.] If $m'_k \mid \mu''_{k-1}$, choose $\gamma_k\in\mathbb{N}$ such that
    \begin{equation}\label{eqn:gamma}
        \gamma_k>\frac{\mu''_{k-1}-1}{m'_kl'_k}
    \end{equation}
    and set
    \begin{align}\label{rule}
    l''_k&\gets l'^2_k\gamma_k,    \nonumber\\
    m''_k&\gets m'_kl'_k\gamma_k+1.
    \end{align}
    \item[5.] Set $\mu''_k= m''_k\mu''_{k-1}$.
    \item[6.] If $k<n_p$, go to step 2.
\end{itemize}
%$p''_k=r_k\exp(i2\pi l''_k/m''_k)$ which satisfy (\ref{eqn:m_notod_mu})
%the set of poles $p''$ which is guaranteed to satisfy (\ref{eqn:m_notod_mu}).
The above procedure terminates in $n_p-1$ iterations. Condition (\ref{eqn:gamma}) ensures that the output sequence $m''$ satisfes (\ref{eqn:m_notod_mu''}), while (\ref{rule}) ensures that the angles $2\pi l''_k/m''$ can tend to $2\pi l'_k/m'$ as follows
\begin{equation}\label{eqn:l"m"e"}
    2\pi l''_k/m''_k=2\pi l'_k/m'_k+\varepsilon''_k, \quad k=1,2,\dots,n_p
\end{equation}
where
$$
\varepsilon''_k=\left\lbrace\begin{array}{ll}
    0 &  m'_k \nmid \mu''_{k-1}\\
    2\pi\left(\frac{ l'^2_k\gamma_k}{m'_kl'_k\gamma_k+1}-\frac{l'_k}{m'_k}\right) & m'_k \mid \mu''_{k-1}.
\end{array}\right.
$$
Since $\lim_{\gamma_k\to+\infty}\vert \varepsilon''_k\vert=0$, one may choose $\gamma_k$ in (\ref{eqn:gamma}) large enough to make $\vert \varepsilon''_k\vert$ as small as desired. By (\ref{eqn:l'm'e'}) and (\ref{eqn:l"m"e"}), we have
$$\vert 2\pi l''_k/m''_k-\theta_k\vert\leq \vert\varepsilon'_k\vert+\vert\varepsilon''_k \vert.$$

We have found angles $2\pi l''_k/m''_k$ that satisfy (\ref{eqn:m_notod_mu''}) and are as close as desired to $\theta_k$ for $k=1,2,\dots,n_p$. Choosing these angles (and their additive inverses) as the angle of the poles in (\ref{eqn:G}) and choosing the same pole magnitudes as in (\ref{eqn:transfer_function}), we obtain a transfer function $G(s)$ that is arbitrarily close to $H(z)$. Since $H(z)$ is an interior point of $\mathcal{M}$, one can ensure $G(z)\in\mathcal{M}$. Therefore,
$G(z)$ is feasible in the linear program~(\ref{eqn:linprog}) by Theorem~\ref{thm:mu}.   
\end{proof}

\subsection{Systems with several positive poles}~\label{sec:several}
As shown in this section, ``most'' systems with one positive pole ($H(z)\in\mathcal{M}$) and no system with two or more positive poles ($H(z)\in\mathcal{X}-\mathcal{M}$) admit positive Markov realizations. Hence, for the linear program~(\ref{eqn:linprog}) to be feasible (and for a positive Markov realization to exist), the system must have one and only one positive pole.

%Nevertheless, one can decompose a system $H(z)$ with $n_+>1$ positive poles into a set of series or parallel subsystems $H_i(z)$ ($i=1,2,\dots,n_+$), each having one positive pole. Then the linear program~(\ref{eqn:linprog}) can be used to obtain a positive Markov realization for each subsystem $H_i(z)$ with dimension $N_i$. After that, a positive realization is obtained for the system $H(z)$, by combining the smaller positive realizations of the subsystems $H_i(z)$. The dimesnion of this ``compound realization'' equals $N=\sum_{i=1}^{n_+}N_i$.

Nevertheless, if a system $H(z)$ with $n_+(\geq 2)$ positive poles is decomposed into $n_+$ series or parallel subsystems each having one positive pole, then the linear program~(\ref{eqn:linprog}) can be applied to each subsystem, giving a positive Markov realization with dimension $N_i$ when feasible, where $i=1,2,\dots,n_+$. These realizations can be combined to obtain a realization for $H(z)$ using standard methods~\cite{combo}. This ``compound'' realization of $H(z)$ is positive and has the dimension $N=\sum_{i=1}^{n_+}N_i$.

\section{Third-order systems}\label{sec:3rdorder}
  For first- and second-order systems, the positive minimal realization dimension equals the system order, \emph{i.e.}, $N^{\star}=n$ where $n\leq 2$~\cite[\S 9]{farinabook}. This is not generally true for higher-order systems. The minimal positive realization problem is still open even for third-order systems~\cite{survey}.
  Therefore, we apply our results to third-order systems in $\mathcal{M}$ and evaluate the sharpness of the upper bound of the minimal positive realization dimension provided by the linear program (\ref{eqn:linprog}).

\begin{figure*}[t]
     \centering
    \begin{subfigure}[b]{0.32\textwidth}
        \centering
        \includegraphics[width=1\linewidth]{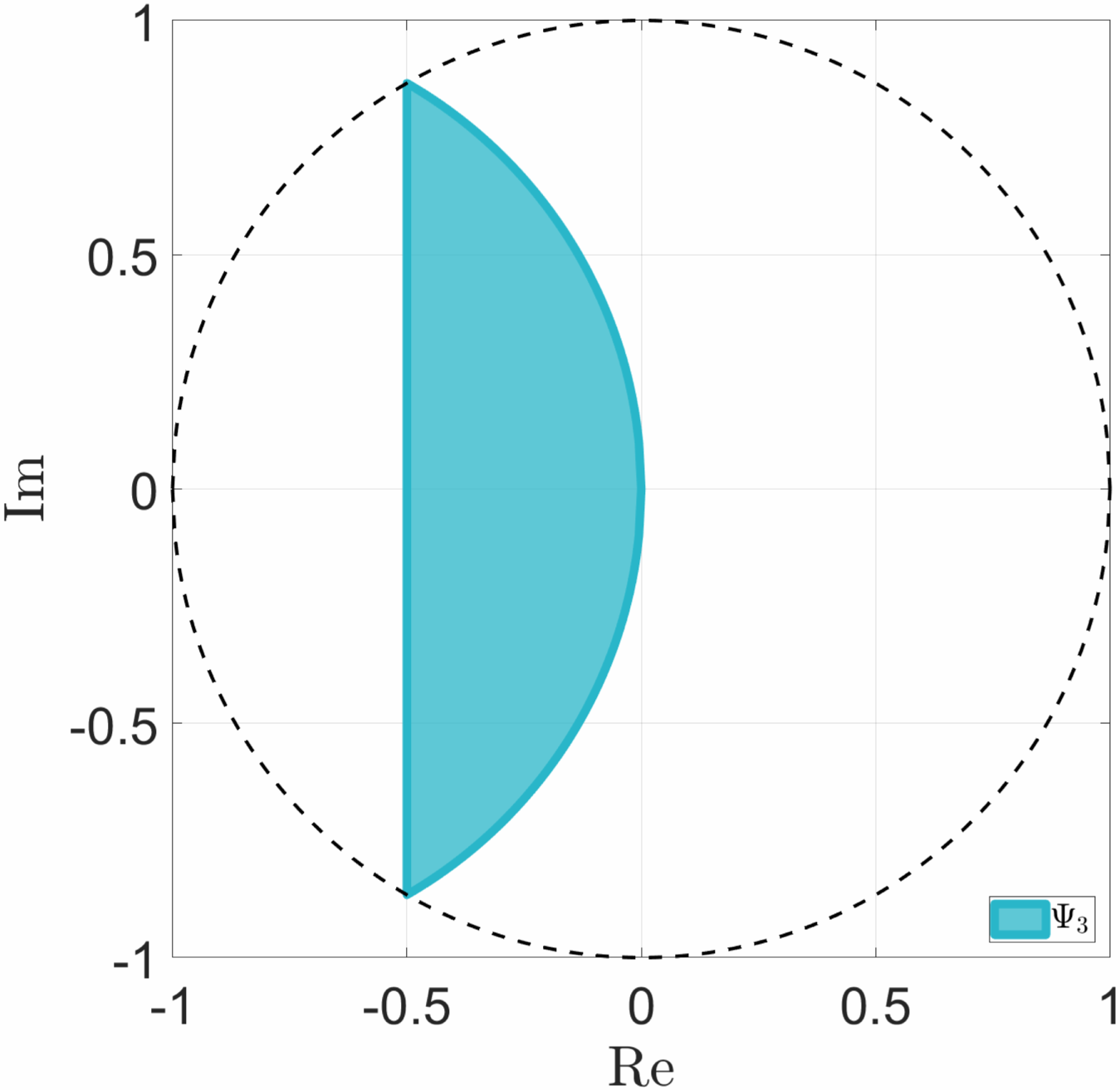}
	    \caption{$N=3$.}
         \label{fig:N3}
     \end{subfigure}
     \begin{subfigure}[b]{0.32\textwidth}
        \centering
        \includegraphics[width=1\linewidth]{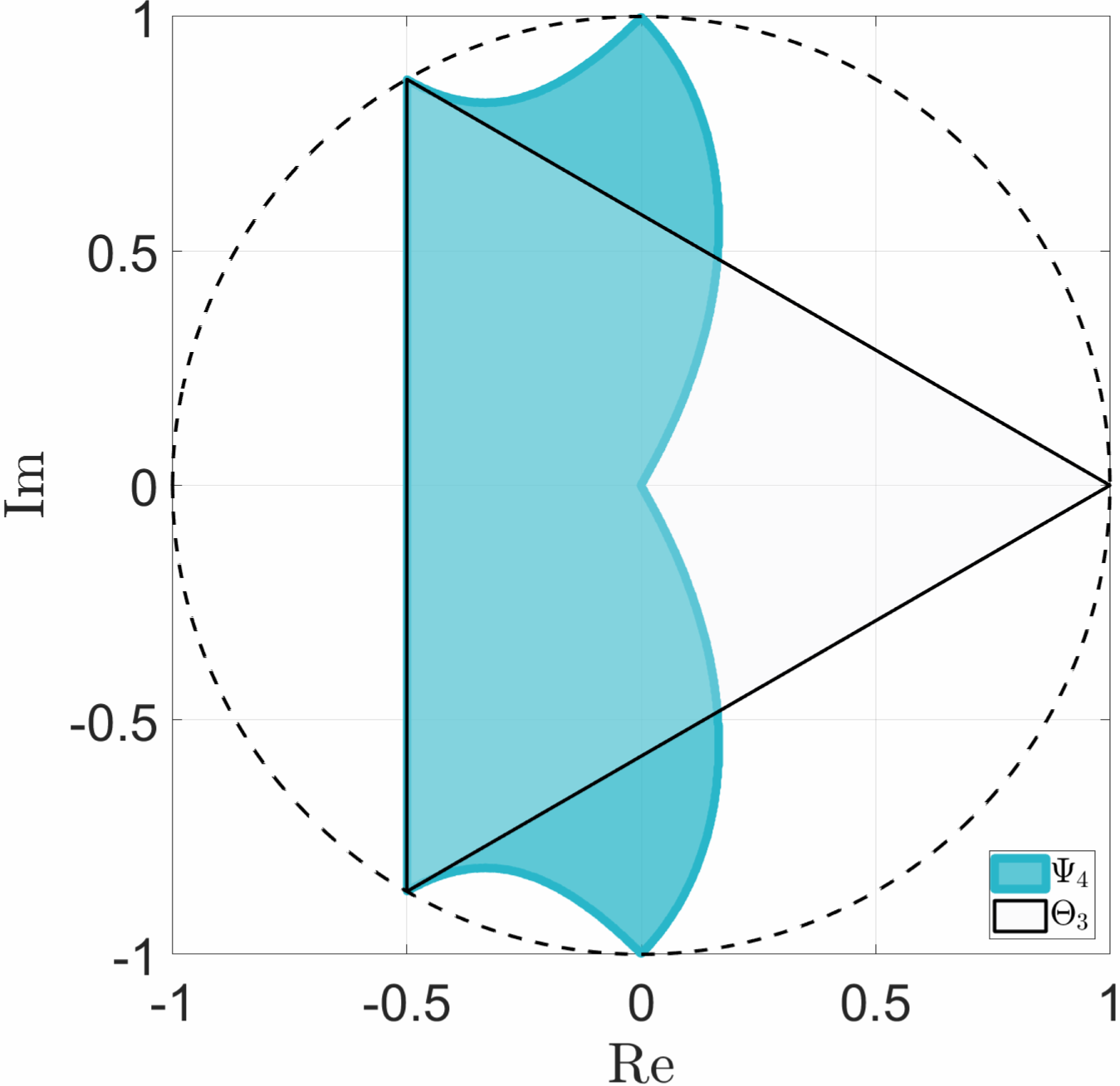}
	    \caption{$N=4$.}
         \label{fig:N4}
     \end{subfigure}
    \begin{subfigure}[b]{0.32\textwidth}
         \centering
        \includegraphics[width=1\linewidth]{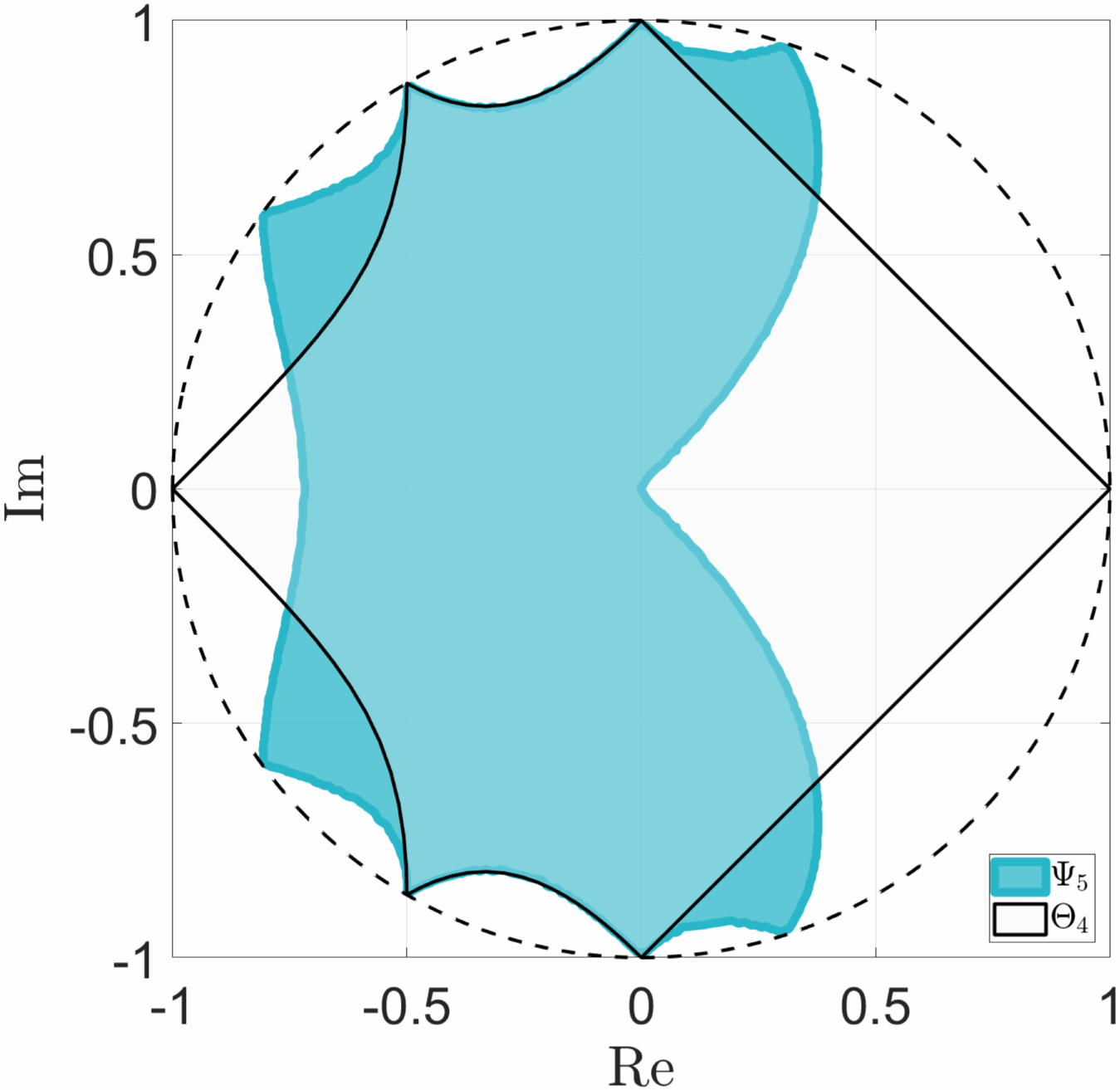}
	    \caption{$N=5$.}
       \label{fig:N5}
    \end{subfigure}
        \caption{When the complex-conjugate poles of a third-order system are located on the blue region, \emph{i.e.}, $p_{2,3}\in\Psi_N$, the system has a positive realization with dimension $N$. If these poles are also \emph{outside} the Karpelevič region enclosed by a solid black boundary, \emph{i.e.}, $p_{2,3}\notin\Theta_{N-1}$, the minimum positive realization dimension equals $N^{\star}=N$.
        }
        \label{fig:N}
\end{figure*}

%Next, we (\ref{eqn:transfer_function})
Consider a third-order system with two complex-conjugate poles $p_{2,3}=x\pm iy$. We use the linear program (\ref{eqn:linprog}) to obtain an upper bound $N$ for the minimal positive realization dimension $N^{\star}$ of this system. To illustrate the results, we define the region on the complex plane where the poles satisfy the condition (\ref{eqn:a*q<0}) as follows:
    \begin{equation}
    \begin{array}{ll}
        \Psi_N&:=\left\lbrace x+iy \,\vert\, \exists q\in\mathbb{N}^{N-n+1} \right.\\
        &\left. \ni q_0=1, (a*q)_k\leq 0, 1\leq k\leq N \right\rbrace,
    \end{array}
    \end{equation}
    %$a=\lbrace 1,-2x-1,2x+x^2+y^2,-x^2-y^2 \rbrace$
    where $a$ is the coefficient sequence of $A(z)$. When $p_{2,3}\in\Psi_N$, the linear program (\ref{eqn:linprog}) is feasible and the system has a positive (Markov) realization of dimension $N$. The regions $\Psi_N$ are plotted in Figure~\ref{fig:N} for $N\in\lbrace 3,4,5\rbrace$. In particular, when the complex-conjugate poles belong to $\Psi_3$, the system has a \emph{minimal} positive realization of order $N=3$. This region $\Psi_3$ is the closure of the region found in \cite{3rd}.

        %For example, choosing $N=3$ gives $q=1$, and thereby, $p_{2,3}\in\Psi_3$ if and only if
    %\begin{align*}
    %    (a*q)_1&=-(2x+1)\leq 0 \\
    %    (a*q)_2&=2x+x^2+y^2\leq 0 \\
    %    (a*q)_3&=-(x^2+y^2)\leq 0 
    %\end{align*}
    %The region $\Psi_3$ specified by these conditions are shown in Figure~\ref{fig:N3}.
    
    As can be seen in Figure~\ref{fig:N}, increasing $N$ results in more relaxed conditions in (\ref{eqn:a*q<0}), and therefore, in larger regions on the complex plane, \emph{i.e.}, $\Psi_{N}\subseteq \Psi_{N+1}$. It is also instructive to compare these regions with Karpelevič regions in Figure~\ref{fig:N}. The Karpelevič region $\Theta_{N}$ characterizes the possible location of an eigenvalue of a non-negative matrix $A\in\mathbb{R}^{N\times N}$ with a unit spectral radius in the complex plane~\cite{ito}. This induces a lower bound on the minimal positive realization dimension $N^{\star}$ based on the location of the system poles as follows~\cite[Theorem 8]{survey}
    \begin{equation}\label{eqn:Karpevich_lowerbound}
        p_{2,3}\not\in\Theta_{N-1} \Rightarrow N\leq N^{\star}.
    \end{equation}
    Therefore, the dimension $N$ found by the linear program (\ref{eqn:linprog}) is minimal, when the poles lie outside the Karpelevič region $\Theta_{N-1}$, \emph{i.e.},
    \begin{equation}\label{eqn:karp_arrow}
        p_{2,3}\in\Psi_N-\Theta_{N-1} \Rightarrow N^{\star}=N.
    \end{equation}
    %We have plotted $\Psi_N$ and $\Theta_{N-1}$ in Figure () for $N=3,4,5$ for example. In the general case when $N\geq 3$,
    For example, consider the case where the complex-conjugate poles are given by
    $$
    p_{2,3}=\vert p_{2,3}\vert \exp(\pm i2\pi l_2/m_2),
    $$
    where $l_2,m_2\in\mathbb{N}$, $l_2<m_2$ and $\textnormal{gcd}(m_2,l_2)=1$. In this case, the linear program (\ref{eqn:linprog}) is feasible with $N=m_2$ regardless of the pole magnitudes, according to Theorem~\ref{thm:mu}. As the magnitude of the poles increases, both poles eventually leave $\Theta_{N-1}$ and the relation (\ref{eqn:karp_arrow}) holds. To see this, note that the Karpelevič region $\Theta_{N-1}$ intersects the unit circle at the points
    \begin{align}\label{eqn:karpoints}
        V_{N-1}&=\Theta_{N-1}\cap \lbrace z\in\mathbb{C}\,\vert\,\vert z\vert=1\rbrace \\
        &= \lbrace \exp(i2\pi l/m)\,\vert\, 0\leq l\leq m\leq N-1,\, \nonumber\\
        &\quad\textnormal{gcd}(m,l)=1\rbrace \nonumber
    \end{align}
    and we have
    $$\lim_{\vert p_{2,3}\vert\to 1} p_{2,3}=\exp(\pm i 2\pi l_2/N)\notin V_{N-1}.$$ %must also have angles with rational multiples of pi
    We conclude that, when the complex-conjugate poles of the third-order system are large enough, the linear program~(\ref{eqn:linprog}) gives the exact minimal positive realization dimension of the system which can be realized in Markov form (\ref{eqn:Markov_realization}). The minimal positive Markov realization dimension equals the minimal positive realization dimension for these systems. This complements the available results for third-order systems in the literature. For example, Corollary~4 in \cite{Pi} asserts that when the complex-conjugate poles belong to the polygons $p_{2,3}\in\Pi_N(\subset \Theta_N)$ in Figure~\ref{fig:PI} and certain additional conditions are satisfied, the system has a positive realization of dimension $N$. The linear program~(\ref{eqn:linprog}) covers the regions missed by $\Pi_N$ at the corners of $\Theta_N$ as shown in Figure~\ref{fig:PI}. We also emphasize that there are other regions where the bound provided by the linear program~(\ref{eqn:linprog}) is larger than that found by \cite{Pi}.\looseness=-1

    \begin{figure}
     \centering
        \includegraphics[width=0.7\linewidth]{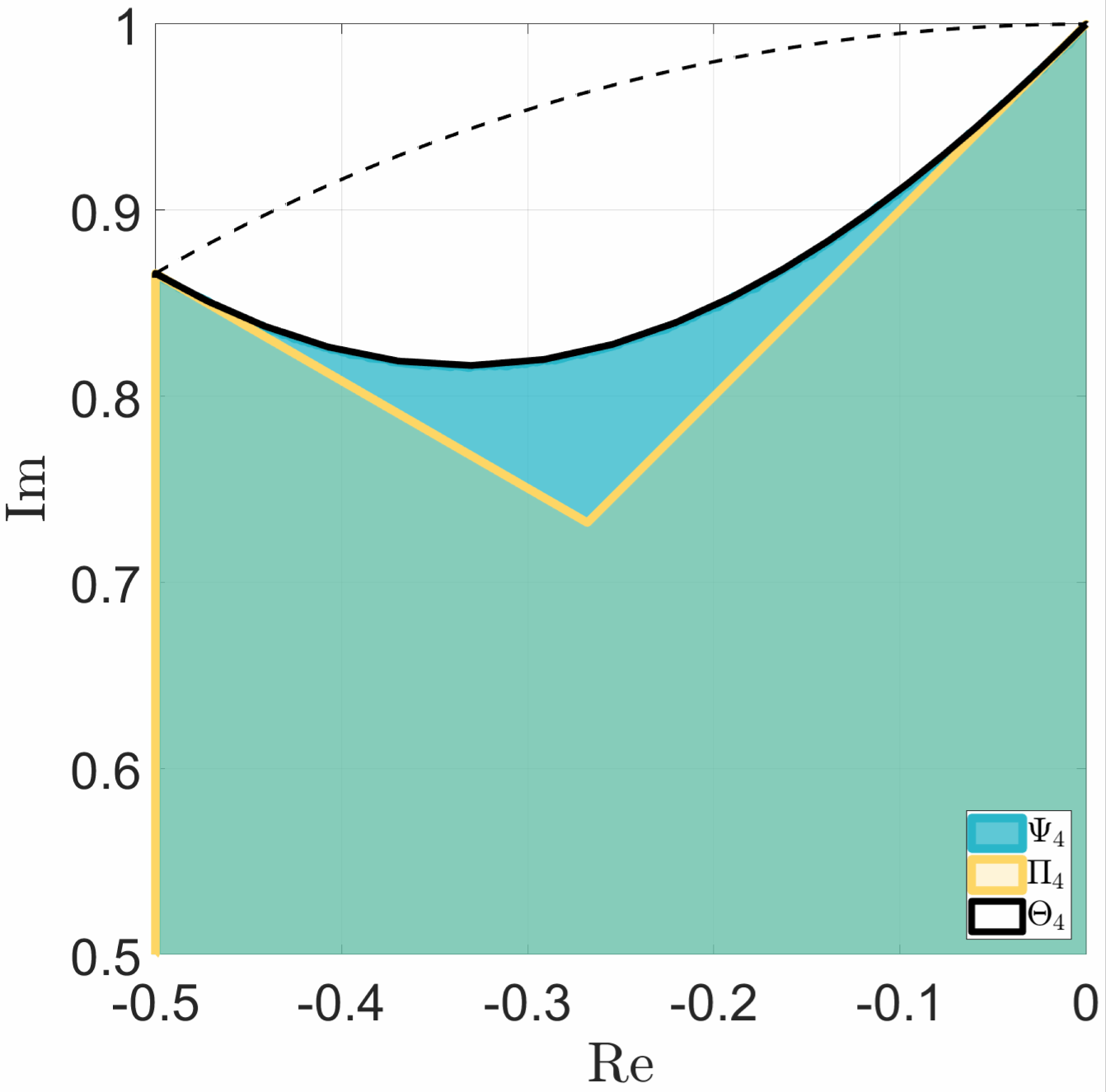}
        \caption{        
        The region $\Pi_4$ (below the yellow line) contains the complex-conjugate poles of third-order systems that admit positive realizations of dimension $4$ according to \cite{Pi}. The linear program~(\ref{eqn:linprog}) still finds a positive realization of dimension $N=4$ outside this region, \emph{i.e.}, $p_{2,3}\in\Psi_4-\Pi_4$. Note that there are no third-order systems with poles outside the Karpelevič region $\Theta_4$ (above the solid black line) with positive realization of dimension $4$ or smaller, according to (\ref{eqn:Karpevich_lowerbound}).\looseness=-1     
        }% $\Pi_N$ contains, as a proper subset,...
        \label{fig:PI}
\end{figure}

\section{Conclusion}\label{sec:conc}
We studied positive Markov realizations and showed that a dense subset of transfer functions with one positive pole admit such realizations (Corollary~\ref{cor:intM}). This result can be extended to systems with several positive poles by a series or parallel connection of these transfer functions (Section~\ref{sec:several}). An interesting feature of positive Markov realization is that one can minimize its dimension using linear programs. The minimum dimension of positive Markov realizations is generally an upper bound of the minimal positive realization dimension. However, these two dimensions are equal for some systems. This includes all third-order systems with complex conjugate poles, where: (i) the pole angles are rational multiples of $\pi$ and (ii) the pole magnitudes are large enough (Section~\ref{sec:3rdorder}).\looseness=-1

%We studied positive Markov realizations and identified a family of transfer functions that admit such realizations. We proposed a linear program to compute them and obtain positive Markov realizations with minimum dimensions. Although the dimension of such realizations is an upper bound of the minimal positive realization dimension in general, they become equal in some cases. 

%For this method to be applicable, the system must have one and only one positive pole. However, a series or parallel connection of several positive sub-systems has a positive realization, whose dimension equals the sum of all the sub-systems dimensions. Hence, it is straightforward to generalize the results to systems with several positive poles.

\section*{Acknowledgment}
The authors would like to thank the anonymous reviewers of a previous paper that brought our attention to this problem.

\printbibliography
\end{document}